%% file: main.tex
\documentclass[11pt]{article}
\usepackage[margin = 2.50cm]{geometry}
\usepackage{microtype}
\usepackage{amsmath,amssymb,amsthm}
\usepackage{graphicx}
\usepackage{tablefootnote}
\usepackage{thmtools, thm-restate}

\usepackage[table]{xcolor} 
\input{support}

\title{
Optimal non-adaptive algorithm for edge estimation
}

\author{
Arijit Bishnu\footnote{Indian Statistical Institute, Kolkata, India}
\and
Debarshi Chanda\footnotemark[1]
\and
Buddha Dev Das\footnotemark[1]
\and 
Arijit Ghosh\footnotemark[1]
\and
Gopinath Mishra\footnote{The Institute of Mathematical Sciences, HBNI, Chennai, India}
}

\date{}

\begin{document}
\maketitle

\begin{abstract}
    We present a simple nonadaptive randomized algorithm that estimates the number of edges in a simple, unweighted, undirected graph, possibly containing isolated vertices, using only degree and random edge queries. For an $n$-vertex graph, our method requires only $\widetilde{O}(\sqrt{n})$ queries, achieving sublinear query complexity. The algorithm independently samples a set of vertices and queries their degrees, and also independently samples a set of edges, using the answers to these queries to estimate the total number of edges in the graph. We further prove a matching lower bound, establishing the optimality of our algorithm and resolving the non-adaptive query complexity of this problem with respect to degree and random-edge queries.
\end{abstract}

\input{1_Introduction}

\input{2_Motivation}

\input{3_Notations}
\input{4_UpperBound}

\input{5_LowerBound}

\section*{Acknowledgements}

Arijit Bishnu acknowledges partial support from the Department of Science and Technology (DST), Government of India, through grant TPN-104427. Arijit Ghosh acknowledges partial support from the Science and Engineering Research Board (SERB), Government of India, through the MATRICS grant MTR/2023/001527, and from the Department of Science and Technology (DST), Government of India, through grant TPN-104427.

\bibliographystyle{alpha}

\bibliography{refs}

\end{document}

%% file: support.tex
\usepackage{float}

\usepackage{amsmath, mathtools, amsthm, amssymb, thm-restate}
\usepackage{algorithm, algpseudocode, yhmath}

\usepackage[textwidth=4em,textsize=tiny]{todonotes}

\theoremstyle{plain}
\newtheorem{theorem}{Theorem}
\newtheorem{lemma}[theorem]{Lemma}

\theoremstyle{definition}

\theoremstyle{remark}

\input{boxes}

\input{macros}

\DeclarePairedDelimiter{\ceil}{\lceil}{\rceil}

\newcommand{\size}[1]{\left|{#1}\right|}

\newcommand{\fbrac}[1]{\left({#1}\right)}
\newcommand{\sbrac}[1]{\left\{{#1}\right\}}
\newcommand{\tbrac}[1]{\left[{#1}\right]}
\newcommand{\abs}[1]{\left|{#1}\right|}

\DeclareMathOperator*{\E}{\mathbb{E}}
\DeclareMathOperator*{\Var}{\mathrm{Var}}

\newcommand{\approxerror}{\varepsilon}

\newcommand{\remove}[1]{}

%% file: boxes.tex
\usepackage[most]{tcolorbox}
\newtcolorbox{idea}[1][]
{
colbacktitle=cyan,
colback=cyan!10,
arc=1pt,
boxrule=1pt,
title=#1 
}

\newtcolorbox{update}[1][]
{
colbacktitle=gray,
colback=gray!10,
arc=1pt,
boxrule=1pt,
title=#1 
}

\newtcolorbox{question}[1][]
{
coltitle=black,
colbacktitle=yellow,
colback=yellow!10,
arc=1pt,
boxrule=1pt,
title=#1 
}

\newtcolorbox{note}[1][]
{
coltitle=black,
colbacktitle=green,
colback=green!10,
arc=1pt,
boxrule=1pt,
title=#1 
}



%% file: macros.tex
\newcommand{\degreeq}{\texttt{Deg}}
\newcommand{\randedgeq}{\texttt{RandEdge}}
\newcommand{\neighbourq}{\texttt{Nbr}}
\newcommand{\edgeexistsq}{\texttt{Pair}}

\newcommand{\graph}{G}
\newcommand{\vertexset}{V}
\newcommand{\edgeset}{E}
\newcommand{\vertexcount}{n}
\newcommand{\edgecount}{m}
\newcommand{\vertex}{v}
\newcommand{\altvertex}{u}

\newcommand{\neighbour}{{\sf N}}
\newcommand{\degree}[1]{\mathrm{deg}\fbrac{{#1}}}

\newcommand{\vertexsamplesize}{K}
\newcommand{\sampledvertices}{S}
\newcommand{\densebucketset}{I_\heavybucket}

\newcommand{\vertexdensity}{\rho}

\newcommand{\heavybucket}{\mathrm{H}}
\newcommand{\lightbucket}{\mathrm{L}}
\newcommand{\bucketset}{\mathrm{B}}
\newcommand{\heavydegapprox}{\widetilde{d}_\heavybucket}
\newcommand{\heavydeg}{\degree{\heavybucket}}

\newcommand{\randomvertices}{T}
\newcommand{\heavyprobapprox}{\widetilde{p}_\heavybucket}
\newcommand{\heavyprob}{p_\heavybucket}

\newcommand{\nonadapalgo}{\texttt{NonAdaptiveEdge}}
\newcommand{\fullnonadapalgo}{\texttt{FullNonAdaptiveEdge}}
\newcommand{\preapprox}{\gamma}
\newcommand{\collisioncount}{r}
\newcommand{\resamplesize}{s}
\newcommand{\collisionedge}{\widetilde{m}_\text{coll}}
\newcommand{\noncollisionedge}{\widetilde{m}_\text{non-coll}}
\newcommand{\iscollision}{\texttt{IsCollision}}
\newcommand{\collision}{k}
\newcommand{\preresamplesize}{t}

\newcommand{\bucket}{B}
\newcommand{\bucketcount}{t}

%% file: 1_Introduction.tex
\section{Introduction}
\label{sec:intro}

Let $\graph = (\vertexset, \edgeset)$ be a simple undirected graph with $\vertexcount = |\vertexset|$ vertices and $\edgecount = |\edgeset|$ edges, possibly containing isolated vertices. For a vertex $\vertex \in \vertexset$, let $\neighbour(\vertex) = \{\altvertex : (\altvertex, \vertex) \in \edgeset\}$ be its neighborhood and $\degree{\vertex} = |\neighbour(\vertex)|$ its degree. We use $[n]$ to denote $\{1, 2, \dots, n\}$.  

Given an approximation parameter $\approxerror > 0$, our goal is to design a \textbf{non-adaptive randomized algorithm} that outputs a $(1\pm\approxerror)$-approximation $\widetilde{\edgecount}$ of $m$ with success probability at least $3/4$.\footnote{We say that $\widetilde{\edgecount}$ is a \emph{$(1 \pm \approxerror)$-approximation} of $\edgecount$ if $|\widetilde{\edgecount} - \edgecount| \leq \approxerror\,\edgecount$.} The graph $\graph$ is accessed through the following two queries:
\begin{itemize}
    \item \emph{Degree query} $\degreeq(\vertex)$: Given a vertex $\vertex$, returns its degree $\degree{\vertex}$.
    
    \item \emph{Random edge query} $\randedgeq$: Returns an edge selected uniformly at random from $\edgeset$.
\end{itemize}

An algorithm is \emph{non-adaptive} if it determines all its queries in advance, based solely on its input parameters and its internal randomness, without depending on the outcomes of prior queries. Otherwise, it is \emph{adaptive}. Non-adaptive algorithms are desirable due to their simplicity and suitability for parallel or offline environments.

In this paper, we present the first (almost) optimal non-adaptive randomized algorithm that estimates $\edgecount$ within a $(1 \pm \approxerror)$ factor using $\widetilde{O}(\sqrt{\vertexcount})$ degree and random-edge queries, where $\widetilde{O}(\cdot)$ suppresses factors polynomial in $\log \vertexcount$ and $1/\approxerror$. Our main contributions are summarized below:
\begin{itemize}
    \item We give the first (almost) optimal {\em non-adaptive algorithm} for this problem.
    
    \item Our algorithm applies even to graphs with {\em isolated vertices}.
    
    \item The algorithm and its analysis are considerably simpler than existing adaptive approaches.
\end{itemize}

\subsection{Motivation}
\label{subsec: Motivation}

We now place our work in the context of prior research on sublinear-time graph parameter estimation~\cite{feige04, avg_deg_danaron08, BerettaT24, beretta2025}. In addition to degree and random edge queries, several other types of queries are used in the property testing literature to access graph structure:

\begin{itemize}    
    \item \emph{Neighbor query} $\neighbourq(\vertex, i)$: Given a vertex $\vertex \in \vertexset$ and index $i \in [n]$, returns the $i$-th neighbor of $\vertex$ if it exists, or $\perp$ otherwise.
    \item \emph{Edge existence query} $\edgeexistsq(\altvertex, \vertex)$: Given two vertices $\altvertex, \vertex \in \vertexset$, returns $1$ if $(\altvertex, \vertex) \in \edgeset$, and $0$ otherwise.
\end{itemize}

The problem of edge estimation, or equivalently, estimation of the average degree under restricted
graph access, was initiated by Feige~\cite{feige04}, who gave a $(2 \pm \approxerror)$-approximation
using $\widetilde{O}(\sqrt{n})$ non-adaptive degree queries and showed that improving this
approximation factor requires $\Omega(n)$ degree queries. This barrier was broken by Goldreich
and Ron~\cite{avg_deg_danaron08}, who used neighbor queries and achieved a $(1 \pm \approxerror)$-approximation
using $\widetilde{O}(\sqrt{n})$ \emph{adaptive} queries. We note that the $\widetilde{O}(\sqrt{n})$ query
complexity of both works holds only when the graph has $\Omega(n)$ edges; otherwise their worst-case query complexity can be as large as $\Omega(n)$.

Their work initiated a line of research on sublinear time algorithms for edge estimation~\cite{eden_lower_bound18, EdenRS19, addanki_et_al:LIPIcs.ESA.2022.2, beretta2025}. Despite their optimal query complexity, these algorithms are adaptive and tend to be more intricate, relying on interleaved degree and neighbor queries across randomly sampled vertices. This motivates the following natural question:

\begin{center}
   \emph{Can we design a non-adaptive randomized algorithm that outputs a $(1 \pm \approxerror)$-approximation\\
   of $\edgecount$ using only $\widetilde{O}(\sqrt{n})$ queries, with a matching lower bound?}
\end{center}

The works of Feige~\cite{feige04} and Goldreich and Ron~\cite{avg_deg_danaron08} suggest that improving edge-estimation guarantees is possible only by enriching the query model. A promising direction is to supplement degree queries with random-edge queries, which have recently found powerful applications in subgraph counting and edge estimation~\cite{Aliakbarpour/Algorithmica/2018/CountingStarSUbgraphsEdgeSampling,AssadiKapralovKhanna/ITCS/2019/SimpleSUblinearSubgraph,beretta2025}.
We study the non-adaptive edge-estimation problem in this strengthened model and prove the following result:

\begin{restatable}{theorem}{upperbound}
\label{theorem:non_adaptive_when_m>cn}
There is a randomized algorithm \nonadapalgo{} that, given $\degreeq$ and $\randedgeq$ query access to a graph $G=(V,E)$ with $|V|=n$, $|E|=m$, and $m \ge n/2$, outputs an estimate $\widetilde{m}$ satisfying
\[
(1-\varepsilon)m \le \widetilde{m} \le (1+\varepsilon)m
\]
using $\widetilde{O}\!\bigl(\,n^{1/2}\varepsilon^{-2.5}\,\bigr)$ queries.
\end{restatable}

We also prove a matching lower bound in a stronger model:

\begin{restatable}{theorem}{lowerbound}
\label{theorem:non-adaptive_lower_bound}
Any non-adaptive randomized algorithm that has access to $\degreeq$, $\neighbourq$, $\edgeexistsq$, and $\randedgeq$ queries and outputs a $2$-approximation of the number of edges requires $\Omega(\sqrt{n})$ queries.
\end{restatable}

\subsection{Notation}

For any $n \in \mathbb{N}$, let $[n] = \{1, 2, \dots, n\}$. We use $\widetilde{O}(\cdot)$ (and similarly, $\widetilde{\Theta}(\cdot)$ and $\widetilde{\Omega}(\cdot)$) to suppress factors polynomial in $\log \vertexcount$ and $1/\approxerror$. An event is said to occur \emph{with high probability} if it occurs with probability at least $1 - \vertexcount^{-c}$ for some constant $c > 0$. We write $a \approx_{\approxerror} b$ (equivalently, $a = (1 \pm \approxerror)b$) to indicate that $a$ is a $(1 \pm \approxerror)$-approximation of $b$. The input graph is $\graph = (\vertexset, \edgeset)$, where $\vertexcount = |\vertexset|$ and $\edgecount = |\edgeset|$.

%% file: 2_Motivation.tex
\section{Prior works on edge estimation}


Table~\ref{table: summary_prior_work} summarises prior results on edge estimation and the {\em adaptivity–query-complexity} trade-offs achieved by existing algorithms.

Feige~\cite{feige04} presented a non-adaptive algorithm that gives a $(2 \pm \approxerror)$-approximation to $\edgecount$ using $\widetilde{O}\fbrac{\sqrt{n} }$ $\degreeq$ queries. He also proved that any algorithm achieving a $(2 - o(1))$-approximation must make $\Omega(n)$ $\degreeq$ queries.

Goldreich and Ron~\cite{avg_deg_danaron08} gave a $(1 \pm \approxerror)$-approximation algorithm with query complexity $\widetilde{O}(\vertexcount/\sqrt{\edgecount})$. Their algorithm uses both $\degreeq$ and random $\neighbourq$ queries adaptively. Their algorithm queries a random neighbour of a vertex, which requires first querying the degree of a vertex, making subsequent queries dependent on prior answers. This makes their algorithm adaptive.

Motwani et al.~\cite{motwani07} presented a $(1 \pm \approxerror)$-approximation algorithm for estimating the sum of values in a set via weighted sampling, but it can be adapted to estimate the number of edges in a graph using $\widetilde{O}(\vertexcount^{1/3})$ $\randedgeq$ and $\degreeq$ queries. The key idea is that sampling a random edge (via $\randedgeq$) and then selecting one of its endpoints yields a vertex sampled with probability proportional to its degree. However, note that simulating the weighted sampling model itself requires the algorithm to be adaptive as the \degreeq{} queries have to be made depending on the output of the \randedgeq{}. Recently,~\cite{BerettaT24} established upper and lower bounds on the algorithm that are tight in terms of all parameters, including $\approxerror$. 

Very recently,~\cite{beretta2025} studied the problem of edge-estimation strictly in the graph context. They considered several models for the queries, and assumes that the graph has no isolated vertices. Firstly, they show that if only \degreeq{}, \neighbourq{}, and  \randedgeq{} queries are available, the query complexity of the problem is $\widetilde{\Theta}(n^{1/3})$. If the \edgeexistsq{} query is also available, then the query complexity of the problem improves to $\widetilde{\Theta}(n^{1/4})$. They also establish tighter bounds when the algorithm is given access to structural queries, and establish bounds with unknown $\vertexcount$. However, all the algorithms they present are adaptive.




\begin{table}
    \centering
\begin{tabular}{ c | c | c | c | c }
 \hline
 \hline
 Queries & Approx & Query Complexity & Adaptivity &  Reference \\ 

 \hline
 \hline
 $\degreeq$ & $2 \pm \approxerror$ & $\vertexcount/ \sqrt{\edgecount}$ & {Yes\tablefootnote{Given a lower bound $\widehat{\edgecount}$ on the number of edges $\edgecount$, the algorithm becomes non-adaptive. But without such a lower bound, the algorithm has to guess a good bound adaptively.}} & \cite{feige04} \\ 
 $\degreeq$ & $1 \pm \approxerror$ & $\vertexcount^{2}/\edgecount$ & Yes & \cite{eden_lower_bound18} \\
 $\degreeq$ + $\neighbourq$ & $1 \pm \approxerror$ & $\vertexcount/\sqrt{\edgecount}$ & Yes & \cite{avg_deg_danaron08} \\ 

  $\randedgeq$ + $\degreeq$ & $1 \pm \approxerror$ & $\vertexcount^{1/3}$ & Yes & \cite{motwani07,BerettaT24,beretta2025} \\

    $\randedgeq$ + $\edgeexistsq$ + $\degreeq$ & $1 \pm \approxerror$ & $\vertexcount^{1/4}$ & Yes & \cite{beretta2025} \\
 \hline
 \rowcolor{green!20}
 $\randedgeq$ + $\degreeq$ & $1 \pm \approxerror$ & $\sqrt{\vertexcount}$ & No & This work \\
 \hline
 \hline
 \end{tabular}
 \caption{\centering{An overview of prior works on edge estimation. The query complexity terms ignores $O(\mathrm{poly}(1/\approxerror,\log\vertexcount))$ terms.}}
 \label{table: summary_prior_work} 
\end{table}

%% file: 4_UpperBound.tex
\section{Overview of the algorithm}
\label{sec:overview_new}

Our non-adaptive algorithm use both \degreeq{} and $\randedgeq$ queries, and builds on the works of Feige~\cite{feige04}, and Goldreich and Ron~\cite{avg_deg_danaron08}. We will first begin by outlining our approach for the case when $\edgecount \geq \vertexcount/2$, and later show how to make this approach work for all values of $\edgecount$.
Recall that we are given an unknown graph $G = (V,E)$ which can be accessed through \degreeq{} and $\randedgeq$ queries, and an approximation parameter $\approxerror$.

We first fix a threshold $\vertexdensity > 0$ which depends only on $\approxerror$ and $\vertexcount$, and a parameter $\gamma = \Theta(\approxerror)$. Suppose we have full access to the whole graph $G$, then we could partition the vertices of the graph into $t = \lceil \log _{1+\gamma} n \rceil +1 $ buckets $B_{0}, \dots, B_{t-1}$ where the degrees of each vertex in $B_{i}$ is between $(1+\gamma)^{i-1}$ and $(1+\gamma)^{i}$. Now consider the partition the set of buckets into two sets $\lightbucket$ (which stands for {\em light} buckets) and $\heavybucket$ (which stands for {\em heavy} buckets). A bucket $B_{i}$ is in the set $\heavybucket$ if 
$$
    |B_{i}| \geq \vertexdensity \vertexcount, 
$$
otherwise $B_{i}$ is put in the set $\lightbucket$. Observe that for sufficiently small $\gamma$ the following approximation holds:
\begin{align}\label{eqn:edgecount_approximation_with_buckets}
    \sum_{i = 1}^{t} |B_{i}| (1+\gamma)^{i-1} \approx_{\approxerror} 2\edgecount,
\end{align}
Let $|\edgeset(\heavybucket)|$, $|\edgeset(\lightbucket)|$, and $|\edgeset(\heavybucket,\lightbucket)|$ denote the number of edges with both endpoints in heavy buckets, with both endpoints in light buckets, and with one endpoint in a light bucket and the other in a heavy bucket, respectively. Therefore,
\begin{align*}
    m = |\edgeset(\heavybucket)| + |\edgeset(\lightbucket)| + |\edgeset(\heavybucket,\lightbucket)|.
\end{align*}
Note that 
\begin{align*}
    \sum_{i: B_{i} \in \lightbucket} |B_{i}| (1+\gamma)^{i-1}  \approx_{\approxerror} 2|\edgeset(\lightbucket)| + |\edgeset(\heavybucket,\lightbucket)|
\end{align*}
One of the crucial ideas of this approach is that one can set the threshold $\vertexdensity$ in such a way that $|\edgeset(\lightbucket)| = o(\widehat{\edgecount})$, and 
\begin{align}\label{eqn:Feige_inequality_1}
    \sum_{j: B_{j} \in \heavybucket} |B_{j}| (1+\gamma)^{i-1} \approx_{\approxerror} 2 |\edgeset(\heavybucket)| + |\edgeset(\heavybucket,\lightbucket)| = 2m-|\edgeset(\heavybucket,\lightbucket)|.
\end{align}
The above equation implies that if we can get a $(1\pm \approxerror)$-approximation $s_{i}$ on the size $|B_{i}|$ of each heavy bucket $B_{i} \in \heavybucket$, then the following quantity will be a $(2 \pm  \approxerror)$-approximation of $\edgecount$:
$$
    \frac{1}{2}\left(\sum_{i: B_{i} \in \heavybucket} s_{i}(1+\gamma)^{i-1} \right).
$$
The above bucketing strategy was developed by Goldreich and Ron~\cite{avg_deg_danaron08}. 

However, we want to get a $(1\pm\approxerror)$-approximation of $m$. Let $\heavydeg$ denote the sum of the degrees of the vertices in the heavy buckets and $\heavyprob := \frac{\heavydeg}{2 \edgecount}$. Observe that 
$$
    \heavydeg = 2|\edgeset(\heavybucket)| + |\edgeset(\heavybucket,\lightbucket)|.
$$
Note that the threshold $\vertexdensity > 0$ for characterizing heavy and light buckets is chosen in such a way that 
$$
    \heavyprob \geq \frac{1}{2} - \Theta{ (\approxerror) }
$$
We have already seen, from Equation~\eqref{eqn:Feige_inequality_1}
\begin{align}
    \sum_{j: B_{j} \in \heavybucket} |B_{j}| (1+\gamma)^{i-1} \approx_{\approxerror} 2 |\edgeset(\heavybucket)| + |\edgeset(\heavybucket,\lightbucket)| = \heavydeg
\end{align}
If we can get a $(1\pm \approxerror)$-approximation $s_{i}$ on the size $|B_{i}|$ of each heavy bucket $B_{i} \in \heavybucket$, then by uniformly sampling $\widetilde{O}\left( \sqrt{\vertexcount}\right)$ vertices and querying their degrees using \degreeq{} queries we can obtain a $(1 \pm  \approxerror)$-approximation of $ \heavydeg $. Let us denote this estimate by $\heavydegapprox$. Now observe that if we can get a $(1 \pm \approxerror)$-approximation of $\widetilde{\heavyprob}$ then ,
$$
    \frac{\heavydegapprox}{\widetilde{\heavyprob}} \approx_{\approxerror} 2m.
$$
We will now show how we can get an $(1 \pm \approxerror)$-approximation of $\widetilde{\heavyprob}$ for $\heavyprob$ using \degreeq{} and $\randedgeq$ queries. Suppose we have access to $\widetilde{O}(1)$ vertices, together with the degrees of the selected vertices, from the graph where each vertex is selected in the sample proportional to its degree. Then using standard averaging arguments and concentration inequalities, we can get an $(1 \pm \approxerror)$-approximation of $\heavyprob$. This can be done by sampling $\widetilde{O}\left( \sqrt{\vertexcount}\right)$ vertices in two different ways and then taking their intersection. The first sample is generated by uniformly sampling $\widetilde{O}\left( \sqrt{\vertexcount}\right)$ vertices and querying their degrees using \degreeq{} queries. The second sample is generated by sampling $\widetilde{O}\left( \sqrt{\vertexcount}\right)$  edges uniformly at random using $\randedgeq$ queries and for each sampled edge, we keep one vertex uniformly at random with probability $1/2$. For a detailed argument, please refer to Section \ref{sec:nonadapt-algo}.

The above argument holds when $m \geq n/2$. For the case when $m \leq n/2$, we can use a standard birthday paradox type argument. The idea is to sample roughly $\widetilde{O}(\sqrt{n})$ random edges uniformly at random and count the number of \emph{collisions}, that is, how many times the same edge appears more than once. The expected number of collisions is roughly $\resamplesize^2 / (2m)$, where $\resamplesize$ is the number of samples. So if we observe even one collision, it suggests that $m = O(n)$. In this sparse regime, the number of collisions provides useful information about the number of edges, and can be used to estimate $m$.

\section{Non-adaptive algorithm using degree and random edges queries}
\label{sec:nonadapt-algo}

We now present a non-adaptive sublinear algorithm \nonadapalgo{} for a $(1 \pm \approxerror)$ estimate of the number of edges in a graph. {Before we start, let us formally define the partition of the vertex set $\vertexset$ into buckets. We define a bucket $\bucket_i$ as:}
\begin{align*}
    \bucket_i = \sbrac{\vertex : \deg(\vertex) \in \left(\fbrac{1 + \preapprox}^{i -1}, \fbrac{1 + \preapprox}^i\right]}
\end{align*}
{Observe that we can partition the vertex set into $\ceil{\log_{1+\preapprox}\vertexcount}$ buckets. We denote $\bucketset = \{\bucket_1,\bucket_2,\ldots,$ $\bucket_{\ceil{\log_{1+\preapprox}\vertexcount}}\}$ to be the set of all buckets. In Section~\ref{sec:overview_new}, we defined the partition of $\bucketset$ into $\heavybucket$ and $\lightbucket$ in terms of the exact size of the buckets, $\size{\bucket_i}$. However, in our algorithm, we cannot make this distinction exactly. Hence, our algorithm decides the partition based on samples, and we show that the decision satisfies the required bounds with high probability. In what follows, the choice of $\heavybucket = \sbrac{\bucket_i \in \bucketset~|~i \in \densebucketset}$ is as decided in line~\ref{alg: non-adap line 3} of Algorithm~\ref{alg: non-adap}, and $\lightbucket$ is defined as $\bucketset \setminus \heavybucket$.}

In Section~\ref{subsec:non-adp non-sparse graph}, we first describe an algorithm \nonadapalgo{} for the graphs where $\edgecount \geq \vertexcount/2$. In Section~\ref{subsec: Complete Algo}, we present an algorithm that works for all graphs.

\subsection{Edge estimation for the case $\edgecount \geq \vertexcount/2$} \label{subsec:non-adp non-sparse graph}

In this section, we present a non-adaptive algorithm assuming $\edgecount \geq \vertexcount/2$. We show the following upper bound.

\begin{theorem}\label{theorem: non_adaptive when m>cn}
    There is a randomized algorithm \nonadapalgo{} that has access to $\degreeq$ and $\randedgeq$ queries to a graph $\graph = (\vertexset, \edgeset)$ with $\size{\vertexset} = \vertexcount$, $\size{\edgeset} = \edgecount$, and $\edgecount\geq \vertexcount/2$. \nonadapalgo{} outputs an estimate $\widetilde{\edgecount}$ of $\edgecount$ satisfying $(1-\approxerror)\edgecount \leq \widetilde{\edgecount} \leq (1 + \approxerror)\edgecount$ using $\widetilde{O}\fbrac{ \frac{\sqrt{\vertexcount}}{\approxerror^{2.5}} }$ queries.
\end{theorem}

\begin{algorithm}
\caption{\nonadapalgo{}}\label{alg: non-adap}
\begin{algorithmic}[1]
\Require \degreeq{} and \randedgeq{} access to the graph $\graph$ with $\edgecount \geq \vertexcount/2$
\State Uniformly and independently select $\widetilde{O}\fbrac{ \frac{\sqrt{\vertexcount}}{\approxerror^{2.5}} }$ vertices from $\vertexset$ and let $\sampledvertices$ be the multi-set of selected vertices. Make \degreeq{} query on vertices in $\sampledvertices$ \label{alg: non-adap line 1}.

\State For $i \in \sbrac{1, 2, \dots, t - 1}$, let $\sampledvertices_i = \sampledvertices \cap \bucket_i$. \label{alg: non-adap line 2}

\State Let $\densebucketset = \sbrac{i : \frac{\size{\sampledvertices_i}}{\size{\sampledvertices}} \geq  \frac{1}{\bucketcount} \sqrt{\frac{\approxerror }{6} \cdot \frac{1}{ \vertexcount} }}$. \label{alg: non-adap line 3}

\State Compute $\heavydegapprox = \frac{\vertexcount}{\size{\sampledvertices}} \sum_{i \in \densebucketset} \size{\sampledvertices_i} (1 + \preapprox)^i$. \label{alg: non-adap line 4}

\State Make $\widetilde{O}\fbrac{\sqrt{\approxerror \vertexcount}}$ \randedgeq{} queries. For each edge, pick one of its endpoint vertices with probability $1/2$ and let $\randomvertices$ be the multi-set of such selected vertices. \label{alg: non-adap line 5}

\State For $i \in \sbrac{0, 1, \dots, t - 1}$, let $\randomvertices_i = \randomvertices \cap \sampledvertices_i$. \label{alg: non-adap line 6}

\State Compute $\heavyprobapprox = \frac{\vertexcount}{\size{\sampledvertices} } \cdot \frac{1}{\size{\randomvertices}} \sum_{i \in \densebucketset} \size{\randomvertices_i}$ \label{alg: non-adap line 7}

\State Output $(1/2) \cdot \heavydegapprox / \heavyprobapprox$ \label{alg: non-adap line 8}
\end{algorithmic}
\end{algorithm}


We begin by sampling a multiset of random vertices $\sampledvertices$ (Line~\ref{alg: non-adap line 1}) and determine the bucket $\bucket_i$ each vertex belongs to using $\degreeq$ queries (Line~\ref{alg: non-adap line 2}). A bucket $\bucket_i$ is classified as heavy if the fraction of sampled vertices from it exceeds a threshold; such indices form the set $\densebucketset$ (Line~\ref{alg: non-adap line 3}). We define the degree sum of the heavy buckets as
\[
\heavydeg = \sum_{i \in \densebucketset} \deg(\bucket_i)
\]


The following lemma establishes a good approximation of the degree sum of heavy buckets. The proof is along the lines of~\cite{avg_deg_danaron08}
\begin{lemma}\label{lemma: approx heavy vertices}
    In line \ref{alg: non-adap line 4} of \nonadapalgo{}, we have $\heavydegapprox = \fbrac{1 \pm \frac{\approxerror}{3}} \heavydeg$ with high probability.
\end{lemma}
\begin{proof}
    Recall from our definition of buckets, 
    for all $i \in [t-1]$, we have
    \[
        \bucket_i = \sbrac{\vertex : \deg(\vertex) \in \left(\fbrac{1 + \preapprox}^{i -1}, \fbrac{1 + \preapprox}^i\right]},   
    \]
    and $\heavydeg =   \sum_{i \in \densebucketset} \deg(\bucket_i)$. Then we have,
\[
\heavydeg \leq  \sum_{i \in \densebucketset} \size{\bucket_i} (1 + \preapprox)^i \leq  (1 + \preapprox)\sum_{i \in \densebucketset} \size{\bucket_i} (1 + \preapprox)^{i-1} \leq (1 + \preapprox) \heavydeg
\]
Let $X_j$ be the indicator random variable which is $1$ if the $j$-th vertex of $\sampledvertices$ is in bucket $\bucket_i$, otherwise it is $0$. Then, $\size{\sampledvertices_i} = \sum_{j\in [\size{\sampledvertices}]} X_j$. Therefore, the expected size of $\size{\sampledvertices_i}$ is
\[
\E\tbrac{\size{\sampledvertices_i}} = \sum_{j\in [\size{\sampledvertices}]} \E\tbrac{X_j} = \size{\sampledvertices} \cdot \frac{\size{\bucket_i}}{\vertexcount}
\]
Let $\vertexdensity = (1/\bucketcount)\sqrt{(\approxerror/8) \cdot (1/\vertexcount)}$, then using the upper tail estimate of the Chernoff bound (see~\cite[Theorem~4.4]{Mitzenmacher_Upfal_2005}), we get that
with high probability, for every $i$ with $\size{\bucket_i} < \vertexdensity\vertexcount$
\[
\frac{\size{\sampledvertices_i}}{\size{\sampledvertices}} \leq \frac{1}{\bucketcount} \sqrt{\frac{\approxerror }{6} \cdot \frac{1}{ \vertexcount} }
\]
The above inequality ensures that with high probability for every $i \in \densebucketset$ we have $\size{\bucket_i} \geq \vertexdensity\vertexcount$.

\paragraph{Upper Estimation:} For $\vertexdensity = (1/\bucketcount)\sqrt{(\approxerror/8) \cdot (1/\vertexcount)}$ and for $\size{\sampledvertices} = \widetilde{O}\fbrac{ \sqrt{\vertexcount} / \approxerror^{2.5} }$, applying the upper estimate of the multiplicative Chernoff bound (see~\cite[Theorem~4.4]{Mitzenmacher_Upfal_2005}) for every $i$ such that $\size{\bucket_i} \geq \vertexdensity\vertexcount$, we have,
\begin{align*}
    \Pr\tbrac{\frac{\size{\sampledvertices_i}}{\size{\sampledvertices}} > \fbrac{1 + \frac{\approxerror}{6}} \frac{\size{\bucket_i}}{\vertexcount}} &\leq \exp{\fbrac{-\frac{\approxerror^2\E[\size{\sampledvertices_i}]}{72}}} \\
    &= \exp{\fbrac{-\frac{\approxerror^2 \cdot \size{\sampledvertices} \cdot \size{\bucket_i}}{72\vertexcount}}} &\text{$\fbrac{\E\tbrac{\size{\sampledvertices_i}} = \size{\sampledvertices} \cdot \frac{\size{\bucket_i}}{\vertexcount}}$}\\
    &\leq \exp{\fbrac{-\frac{\approxerror^2\cdot\size{\sampledvertices}\cdot\vertexdensity}{72}}}&\text{$\fbrac{\size{\bucket_i} \geq \vertexdensity\vertexcount}$}\\
    &\leq \frac{1}{\log \vertexcount \cdot \vertexcount^c}
    &\mbox{$\fbrac{\size{\sampledvertices} = \widetilde{O}\fbrac{\frac{1}{\vertexdensity\approxerror^2}} = \widetilde{O}\fbrac{{\frac{\sqrt{\vertexcount}}{\approxerror^{2.5}}}}}$}
\end{align*}
Therefore, with high probability we have,
\begin{align}\label{ineq: upper bound}
    \frac{\size{\sampledvertices_i}}{\size{\sampledvertices}} \leq \fbrac{1 + \frac{\approxerror}{6}} \frac{\size{\bucket_i}}{\vertexcount}
\end{align}
Now, summing the inequality~\eqref{ineq: upper bound} for heavy buckets $\bucket_i$ where $i \in \densebucketset$, using the union bound we get with high probability,
\begin{align}
    \heavydegapprox = \frac{\vertexcount}{\size{\sampledvertices}} \sum_{i \in \densebucketset} \size{\sampledvertices_i} (1 + \preapprox)^i &\leq  \sum_{i \in \densebucketset} \fbrac{1 + \frac{\approxerror}{6}} \size{\bucket_i} (1 + \preapprox)^{i}
    \leq \fbrac{1 + \frac{\approxerror}{6}} (1 + \preapprox) \heavydeg
\end{align}
Taking, $\preapprox = \approxerror/10$, we get $\heavydegapprox \leq \fbrac{1 + \dfrac{\approxerror}{3}}\heavydeg$.

\paragraph{Lower Estimation:} Note that by applying the lower tail estimate of the multiplicative Chernoff bound (see \cite[Theorem~4.5]{Mitzenmacher_Upfal_2005}) in a similar way as above, for every $i$ with $\size{\bucket_i} \geq \vertexdensity\vertexcount$, then with high probability we get
\begin{align}\label{ineq: lower bound}
    \frac{\size{\sampledvertices_i}}{\size{\sampledvertices}} \geq \fbrac{1 - \frac{\approxerror}{3}} \frac{\size{\bucket_i}}{\vertexcount}
\end{align}
Therefore, summing the inequality~\eqref{ineq: lower bound} over the heavy buckets, using the union bound, we get, with high probability,
\begin{align}
    \heavydegapprox = \frac{\vertexcount}{\vertexsamplesize} \sum_{i \in \densebucketset} \size{\sampledvertices_i} (1 + \preapprox)^i &\geq  \sum_{i \in \densebucketset} \fbrac{1 - \frac{\approxerror}{3}} \size{\bucket_i} (1 + \preapprox)^{i}
    \geq \fbrac{1 - \frac{\approxerror}{3}}  \heavydeg
\end{align}
This completes the proof.
\end{proof}

Recall that we have defined $\heavyprob = \frac{\heavydeg}{2 \edgecount}$. We intend to get a good estimate of $\heavyprob$ using $\randomvertices$. For that purpose, we require a lower bound on the $\heavyprob$, which we establish in the following lemma.

\begin{lemma} \label{lemma: probability of edges}
    With high probability, we have $\heavyprob  \geq \frac{1}{2} - \frac{\approxerror}{8}$.
\end{lemma}

\begin{proof}


    Using the upper tail estimate of the Chernoff bound (see~\cite[Theorem~4.4]{Mitzenmacher_Upfal_2005}) for every $i$ such that $\frac{\size{\bucket_i}}{\vertexcount} \geq \frac{1}{\bucketcount} \sqrt{\frac{\approxerror }{2\vertexcount}}$ with high probability,
    \[
    \frac{\size{\sampledvertices_i}}{\size{\sampledvertices}} \geq \frac{1}{\bucketcount} \sqrt{\frac{\approxerror }{6\vertexcount}  }.
    \]
     This ensures that, if $i \notin \densebucketset$ with high probability we have $\size{\bucket_i} \leq \frac{1}{\bucketcount} \sqrt{\frac{\approxerror \vertexcount}{2} }$. Now, we can upper bound the number of vertices in $\lightbucket$ using the union bound. Therefore, with high probability,
    \begin{equation} \label{ineq: light vertex upper bound}
        |\lightbucket| \leq \bucketcount \cdot \frac{1}{\bucketcount}\sqrt{\frac{\approxerror \vertexcount}{2}} = \sqrt{\frac{\approxerror \vertexcount}{2}}
    \end{equation}
    Recall that $|\edgeset(\heavybucket)|$, $|\edgeset(\lightbucket)|$, and $|\edgeset(\heavybucket,\lightbucket)|$ denote the number of edges whose both endpoints lie in $\heavybucket$, both endpoints lie in $\lightbucket$, and one endpoint each lie in $\heavybucket$ and $\lightbucket$, respectively. Then, we have $\edgecount = |\edgeset(\heavybucket)| + |\edgeset(\heavybucket, \lightbucket)| + |\edgeset(\lightbucket)|$, and $\heavydeg = 2|\edgeset(\heavybucket)| + |\edgeset(\heavybucket, \lightbucket)|$. Then, we have:
    \begin{align*}
        \heavydeg &= 2|\edgeset(\heavybucket)| + |\edgeset(\heavybucket, \lightbucket)|
        \geq |\edgeset(\heavybucket)| + |\edgeset(\heavybucket, \lightbucket)|
        \geq \edgecount - \frac{\approxerror\vertexcount}{2} \geq \fbrac{1 - \frac{\approxerror}{4}}\edgecount
    \end{align*}
    Where the last inequality uses the fact that $\edgecount \geq \vertexcount/2$. By the definition of $\heavyprob$, dividing both sides by $2\edgecount$ completes the proof. 
\end{proof}

In Lemma~\ref{lemma: probability of edges}, we established that $\heavyprob \geq \frac{1}{2} - \frac{\approxerror}{8}$. This lower bound ensures that the contribution of edges to heavy buckets is sufficiently large. As a result, we can obtain a good multiplicative approximation of $\heavyprob$ using random edge samples. We formalise this in the following lemma.
\begin{lemma} In line \ref{alg: non-adap line 7} of \nonadapalgo{}, we have $\heavyprobapprox = \fbrac{1 \pm \frac{\approxerror}{3}} \heavyprob$ with high probability.
\end{lemma}
\begin{proof}
    For a vertex $\vertex$, if $\vertex \in \randomvertices$, then $\vertex$ is an endpoint vertex of a sampled random edge selected with probability $1/2$. Therefore,
    \[
    \Pr\tbrac{\vertex \in \randomvertices} = \frac{\deg(\vertex)}{2\edgecount}
    \]
    Given $\randomvertices_i = \randomvertices \cap \sampledvertices_i$, then the expected size of $\randomvertices_i$ is,
    \begin{align}
        \E \tbrac{\size{\randomvertices_i}} = \size{\randomvertices} \cdot \frac{\sum_{v \in \bucket_i} \deg(v) }{2\edgecount} \cdot \frac{\size{\sampledvertices}}{\vertexcount}
        = \size{\randomvertices} \cdot \frac{\deg\fbrac{\bucket_i}}{2\edgecount} \cdot \frac{\size{\sampledvertices}}{\vertexcount}
    \end{align}
Using the linearity of expectation, we get 
\begin{align}\label{eq: expected random edges}
    \E \tbrac{\sum_{i \in \densebucketset} \size{\randomvertices_i}} = \size{\randomvertices} \cdot \frac{\deg\fbrac{\heavybucket}}{2\edgecount} \cdot \frac{\size{\sampledvertices}}{\vertexcount}
    = \size{\randomvertices} \cdot \heavyprob \cdot \frac{\size{\sampledvertices}}{\vertexcount}
\end{align}
In Line~\ref{alg: non-adap line 7} of \nonadapalgo{}, 
\begin{align*}
    \heavyprobapprox = \frac{\vertexcount}{\size{\sampledvertices} } \cdot \frac{1}{\size{\randomvertices}} \sum_{i \in \densebucketset} \size{\randomvertices_i} 
\end{align*}
Therefore, the expectation of $\heavyprobapprox$ is
\begin{align*}
    \E[\heavyprobapprox] = \frac{\vertexcount}{\size{\sampledvertices} } \cdot \frac{1}{\size{\randomvertices}} \E\tbrac{\sum_{i \in \densebucketset} \size{\randomvertices_i}} = \heavyprob
\end{align*}
Given $\heavyprob \geq \frac{1}{2} - \frac{\approxerror}{8}$ (Lemma~\ref{lemma: probability of edges}), by applying the two-sided Chernoff bound (see \cite[Corollary~4.6]{Mitzenmacher_Upfal_2005}) over $\size{\randomvertices}$ random edge samples ensures that with high probability
\begin{align*}
    \Pr\tbrac{\abs{\heavyprobapprox - \E[\heavyprobapprox]} > \frac{\approxerror}{3}\E[\heavyprobapprox]} &\leq 2\exp{\fbrac{-\frac{\approxerror^2\cdot\E\tbrac{\sum_{i \in \densebucketset} \size{\randomvertices_i}}}{27}}}\\
    &= 2 \exp{\fbrac{-\frac{\approxerror^2 \cdot \size{\randomvertices}\cdot \size{\sampledvertices} \cdot \heavyprob }{27\vertexcount}}} &\text{(From Equation~\eqref{eq: expected random edges})}\\
    &\leq 2 \exp{\fbrac{-\frac{\approxerror^2 \cdot \size{\randomvertices}\cdot \size{\sampledvertices} \cdot 0.4 }{27\vertexcount}}} &\text{$\fbrac{\heavyprob \geq 0.4 \text{ for } \approxerror \leq 0.8 }$}\\
    &\leq \frac{2}{\vertexcount^c} &\text{$\fbrac{\size{\randomvertices} = \widetilde{O}\fbrac{\frac{\vertexcount}{\approxerror^2\size{\sampledvertices}}} = \widetilde{O}\fbrac{\sqrt{\approxerror\vertexcount}}}$}
\end{align*}
So, with high probability, we have,
\begin{align}
    \heavyprobapprox =  \fbrac{1 \pm \frac{\approxerror}{3}} \heavyprob
\end{align}
\end{proof}

Now, from Lemmas~\ref{lemma: approx heavy vertices} and~\ref{lemma: probability of edges}, we have multiplicative approximations to both $\heavydeg$ and $\heavyprob$. Recall that the true edge count satisfies
\[
\edgecount = \frac{\heavydeg}{2\heavyprob}.
\]
Using the approximations to $\heavydeg$ and $\heavyprob$, we can construct an estimator for $\edgecount$ that is accurate up to a $(1 \pm \approxerror)$ multiplicative factor. We formalise this result in the proof for Theorem~\ref{theorem: non_adaptive when m>cn}.

\begin{proof}[{Proof of Theorem~\ref{theorem: non_adaptive when m>cn}}]
    The estimator for the number of edges is given by
    \[
    \hat{m} := \frac{\heavydegapprox}{2\heavyprobapprox}
    \]
    From Lemmas~\ref{lemma: approx heavy vertices} and~\ref{lemma: probability of edges}, we have
    $\heavydegapprox = \fbrac{1 \pm \dfrac{\approxerror}{3}} \heavydeg$ and $\heavyprobapprox = \fbrac{1 \pm \dfrac{\approxerror}{3}} \heavyprob$.
    Applying the union bound, these approximations hold simultaneously with high probability. Thus,
    \[
    \frac{ \fbrac{1 - \frac{\approxerror}{3}} }{ \fbrac{1 + \frac{\approxerror}{3}} } \cdot \frac{\heavydeg}{2\heavyprob}
    \leq \frac{\heavydegapprox}{2\heavyprobapprox}
    \leq \frac{ \fbrac{1 + \frac{\approxerror}{3}} }{ \fbrac{1 - \frac{\approxerror}{3}} } \cdot \frac{\heavydeg}{2\heavyprob}
    \]
    Using the inequalities $1 - \approxerror \leq \dfrac{1 - \frac{\approxerror}{3}}{1 + \frac{\approxerror}{3}}$ and $\dfrac{1 + \frac{\approxerror}{3}}{1 - \frac{\approxerror}{3}} \leq 1 + \approxerror$, for any $\approxerror \in (0,1) $, we obtain,
    \[
    (1 - \approxerror) \cdot \edgecount
    \leq \frac{\heavydegapprox}{2\heavyprobapprox}
    \leq (1 + \approxerror) \cdot \edgecount
    \]
    Therefore, with high probability, \nonadapalgo{} gives an \((1 \pm \approxerror)\)-factor approximate of \(\edgecount\).
\end{proof}

\subsection{Complete algorithm}\label{subsec: Complete Algo}

In the previous section, we presented a non-adaptive algorithm \nonadapalgo{} under the assumption that $\edgecount \geq \vertexcount/2$. We now describe how to remove this assumption using a simple collision-based test inspired by the birthday paradox.

Suppose we independently sample $\resamplesize = \Theta(\sqrt{\vertexcount})$ edges uniformly at random. Let $\collisioncount$ be the number of collisions among these samples. If $\edgecount \leq \vertexcount$, then with probability, a collision occurs, and the number of collisions is used to estimate $\edgecount$. If $\edgecount \geq 8\vertexcount$, then we proceed to run \nonadapalgo{}, which works in this regime.

\begin{algorithm}
\caption{\iscollision{}}
\label{alg: collision}
\begin{algorithmic}[1]
\Require \randedgeq{} access to the graph $\graph$
\For {$i = 1$ to $O(\log \vertexcount)$} \label{alg: collision line 1}
    \State  Uniformly and independently sample $\preresamplesize = \sqrt{2\vertexcount} $ random edges. \label{alg: collision line 2}
    \State If there is a collision, then $\collision_i = 1$, else $\collision_i = 0$. \label{alg: collision line 3}
\EndFor
\State If more than half of $\collision_i$ are $1$, then output $1$, else output $0$.
\end{algorithmic}
\end{algorithm}

\begin{lemma}\label{lemma: collision detection}
    With high probability, the following statements hold:
    \begin{enumerate}
    \item If $\edgecount \leq  \vertexcount$, then $\iscollision{}$ outputs $1$.
    \item If $\edgecount \geq 8\vertexcount$, then $\iscollision{}$ outputs $0$.
\end{enumerate}
\end{lemma}
\begin{proof}
    For the proof, we use the birthday paradox techniques \cite[Chapter 5]{Mitzenmacher_Upfal_2005}. Consider the $i$-th iteration (Line~\ref{alg: collision line 1} of \iscollision{}). The probability of no collision in the $i$-th iteration is given by,
    \begin{align*}
        \Pr[ \collision_i = 0 ] = \prod_{ j = 0 }^{ \preresamplesize - 1 } \left( 1 - \frac{ j }{ \edgecount } \right)
    \end{align*}

    If $\edgecount \leq  \vertexcount$, then 
    \begin{align*}
        \Pr[ \collision_i = 0 ] \leq \exp \left( - \sum_{ j = 0 }^{ \preresamplesize - 1 } \frac{j}{ \edgecount } \right) = \exp \left( - \frac{ \preresamplesize ( \preresamplesize - 1 )}{2\edgecount} \right) \leq \exp\left(-\frac{ \preresamplesize ( \preresamplesize - 1 )}{ 2 \vertexcount } \right)
    \end{align*}
    Here, we have used $ 1 - x \leq e^{-x}$. So, for large $ \vertexcount $ and $ \preresamplesize = \sqrt{ 2 \vertexcount } $, we have
    \begin{align*}
        \Pr[ \collision_i = 0 ] \leq \frac{ 1 }{ e }
    \end{align*}

    If $ \edgecount \geq 8 \vertexcount $, then 
    \begin{align*}
        \Pr[ \collision_i = 0 ] \geq \exp \fbrac{ - \sum_{ j = 0 }^{ \preresamplesize - 1 } \fbrac{ \frac{ j }{ \edgecount } + \frac{ j^2 }{ \edgecount^2 } } } \geq \exp \left( - \frac{ \preresamplesize ( \preresamplesize - 1 )}{ 16 \vertexcount } + \Theta \fbrac{ \frac{ 1 }{ \sqrt{ \vertexcount } } }\right)
    \end{align*}
    Here, we have used $ 1 - x \geq e^{ -x -x^2 }, \; x \leq 0.5 $. So, for large $ \vertexcount $ and $ \preresamplesize = \sqrt{ 2 \vertexcount } $, we have
    \begin{align*}
        \Pr[ \collision_i = 0 ] \geq \frac{ 1 }{ 2 }
    \end{align*}

    Now, when $ \edgecount \leq \vertexcount $, $ \Pr[ \collision_i = 0 ] \leq \frac{ 1 }{ e } $. \iscollision{} outputs $0$ when more than half of $ \collision_i = 0 $. Over $ O( \log \vertexcount ) $ parallel runs, using the upper tail estimate of the Chernoff Bound (see \cite[Theorem~4.4]{Mitzenmacher_Upfal_2005}), the probability that \iscollision{} outputs $0$ is at most $ \exp ( - \log \vertexcount) = \frac{1}{\vertexcount}$. Similarly, $ \edgecount \geq 8\vertexcount $, the probability that \iscollision{} outputs $1$ is at most $ \frac{1}{\vertexcount}$.
\end{proof}

\begin{algorithm}
\caption{\fullnonadapalgo{}}
\label{alg: fully-non-adap}
\begin{algorithmic}[1]
\Require \degreeq{} and \randedgeq{} access to the graph $\graph$

\State Run \iscollision{}. Let the output be $\collision$.\label{alg2: Line 1}

\State  Uniformly and independently sample $ \resamplesize = \Theta \fbrac{ \frac{ \sqrt { \vertexcount } \log \vertexcount }{ \approxerror } }  $ random edges. Let the number of collisions be $ \collisioncount $.\label{alg2: Line 2}

\State Run \nonadapalgo{}. Let the output be $\noncollisionedge$.\label{alg2: Line 3}

\State Output:
\[
\widetilde{m} = 
\begin{cases}
\collisionedge = \binom{\resamplesize}{2}/\collisioncount, & \text{if $\collisioncount > 0$ and $ \collision = 1 $} \\
\noncollisionedge, & \text{Otherwise }
\end{cases}
\]

\end{algorithmic}
\end{algorithm}

Note that we make all the queries regarding our algorithm at the start and then make decisions based on the queries. While the output depends on the outcome of these intermediate steps, all the queries are made a priori.

\begin{lemma}
\label{lemma: complete-non-adap}
With high probability, the following statements hold:
\begin{enumerate}
    \item If $\edgecount \leq  \vertexcount$, then \fullnonadapalgo{} outputs
    \[
    \collisionedge \in (1 \pm \approxerror) \cdot \edgecount.
    \]
    \item If $\edgecount \geq 8\vertexcount$, then \fullnonadapalgo{} outputs.
    \[
    \noncollisionedge \in (1 \pm \approxerror) \cdot \edgecount
    \].
\end{enumerate}
\end{lemma}
\begin{proof}
    When \( \edgecount \leq \vertexcount \), from Lemma~\ref{lemma: collision detection} with high probability, a collision is observed.
    
    Let $\sbrac{e_1, e_2, \ldots, e_{\resamplesize}}$ be the sampled edges and $X_{ij}$ be the indicator random variable which is $1$ if the $e_i$ is the same as $e_j$. Since, $\Pr[X_{ij}=1] = \frac{1}{\edgecount}$, we have $\E[X_{ij}] = \frac{1}{\edgecount}$.  
Thus, the number of collisions is defined by
\begin{align*}
    X = \sum_{1 \le i < j \le \resamplesize} X_{ij},
\end{align*}
Using the linearity of expectation, we get
\begin{align*}
\E[X] = \binom{\resamplesize}{2}\cdot \frac{1}{\edgecount}
\end{align*}
and we have:
\begin{align*}
    \Var(X) &= \sum_{i < j } \Var(X_{ij}) + 2\sum_{i < j, k < l} \mathrm{Cov}(X_{ij},X_{kl})
\end{align*}
As, $X_{ij}$ is a $0,1$-random variable, we also have $\Var(X_{ij}) \leq \E[X_{ij}^2] = \frac{1}{\edgecount}$. We get
\begin{align*}
    \sum_{i < j } \Var(X_{ij}) \leq \binom{\resamplesize}{2} \cdot \frac{1}{\edgecount} =  \E[X]
\end{align*}
Now, notice that if $X_{ij}$ and $X_{k\ell}$ are disjoint pairs (no indices overlap), then they are independent. If $X_{ij}$ and $X_{k\ell}$ are overlapping pairs (share exactly one index, e.g. $j = k$), then
\begin{align*}
    \mathrm{Cov}(X_{ij},X_{j\ell}) &= \E[X_{ij}X_{j\ell}] - \E[X_{ij}]\E[X_{j\ell}]\\  
    &= \Pr[X_{ij} = X_{j\ell} = 1] - \frac{1}{\edgecount^2}\\
    &= \Pr[e_i = e_j = e_\ell] - \frac{1}{\edgecount^2} \\
    &= 0
\end{align*}
Therefore,
\begin{align*}
    \Var(X)  = \sum_{i < j } \Var(X_{ij}) \leq \E[X]
\end{align*}
Now, applying Chebyshev’s inequality (see~\cite[Theorem~3.6]{Mitzenmacher_Upfal_2005}), we get
\begin{align*}
    \Pr[\abs{X - \E[X]} > \frac{\approxerror}{3} \E[X]] &\leq \frac{9\Var(X)}{\approxerror^2 (\E[X])^2}\\
    &\leq \frac{9}{\approxerror^2\E[X]} &\text{$\fbrac{\Var(X) \leq \E[X]}$}\\
    &= \frac{18\edgecount}{\approxerror^2\resamplesize(\resamplesize - 1)} &\text{$\fbrac{\E[X] = \binom{\resamplesize}{2}\cdot \frac{1}{\edgecount}}$}\\
    &\leq \frac{1}{ \log^2 \vertexcount } &\text{($\edgecount \leq \vertexcount$ and $\resamplesize = \Theta(\sqrt{ \vertexcount }\log \vertexcount /{\approxerror} )$)}
\end{align*}
We can boost the probability further by running the algorithm $O(\log \vertexcount)$ times in parallel. Now, with high probability, we have 
\begin{align}
    \fbrac{1 - \frac{\approxerror}{3}}\E\tbrac{X} \leq X \leq \fbrac{1 + \frac{\approxerror}{3}} \E\tbrac{X}
\end{align}
which implies
\begin{align}
    \fbrac{1 - \frac{\approxerror}{3}} \binom{\resamplesize}{2} \cdot \frac{1}{\edgecount} \leq X \leq \fbrac{1 + \frac{\approxerror}{3}} \binom{\resamplesize}{2} \cdot \frac{1}{\edgecount}
\end{align}
and
\begin{align}
     \frac{\edgecount}{\fbrac{1 + \frac{\approxerror}{3}}}\leq \binom{\resamplesize}{2} \cdot \frac{1}{X} \leq \frac{\edgecount}{\fbrac{1 - \frac{\approxerror}{3}}}   
\end{align}
Using the fact that for any $\approxerror \in (0, 1)$
    \begin{equation*}
        1 - \approxerror \leq \frac{1}{\fbrac{1 + \frac{\approxerror}{3}}} \;\text{ and }\; \frac{1}{\fbrac{1 - \frac{\approxerror}{3}}} \leq 1 + \approxerror
    \end{equation*}
Recall $X$ is the number of collisions. Therefore, we have
    \begin{align}
        \fbrac{1 - \approxerror}\cdot \edgecount \leq \binom{\resamplesize}{2} \cdot \frac{1}{\collisioncount} \leq \edgecount \cdot \fbrac{1 + \approxerror}
    \end{align}
The proof of the first part is completed.

For the second part, when $\edgecount \geq 8 \vertexcount$, from Lemma~\ref{lemma: collision detection} with high probability, no collision is observed, and we proceed with the output of \nonadapalgo{}, which completes the proof of the second part.
\end{proof}

\begin{lemma}\label{lem:correctness}
\fullnonadapalgo{} returns $(1\pm \approxerror)$ to the number of edges $m$ with high probability.   
\end{lemma}
\begin{proof}
By Lemma~\ref{lemma: complete-non-adap} we already have the following guarantees:
\begin{itemize}
  \item If $m \le n$ then \fullnonadapalgo{} (the collision branch) returns a $(1\pm\varepsilon)$-approximation to $m$.
  \item If $m \ge 8n$ then \fullnonadapalgo{} (the non-collision branch) returns a $(1\pm\varepsilon)$-approximation to $m$.
\end{itemize}

If $\vertexcount < \edgecount < 8\edgecount$, then both the collision estimator and the heavy-regime estimator succeed with high probability when $ \edgecount $ is a constant-factor multiple of $ \vertexcount $. When $ \edgecount < 8 \vertexcount $, then using the same approach as in Lemma~\ref{lemma: complete-non-adap} we can show that the collision-based estimator $ \collisionedge $ is a $(1\pm \approxerror)$-approximation with high probability. When $ \edgecount > \vertexcount $, then the output of $ \nonadapalgo{} $ is $(1\pm \approxerror)$-approximation. Since, $ \nonadapalgo{} $ works in the regime $ \edgecount \geq \vertexcount/2 $.
%
%
\end{proof}

\upperbound*
\begin{proof}
     From Lemma~\ref{lem:correctness}, \fullnonadapalgo{} outputs an estimate $\widetilde{\edgecount}$ of $\edgecount$ satisfying $\widetilde{\edgecount} = (1 \pm \approxerror)\edgecount$ with high probability.
     Also, from Theorem~\ref{theorem: non_adaptive when m>cn}, the number of queries made by \nonadapalgo{} is $\widetilde{O}\fbrac{\frac{\sqrt{\vertexcount}}{\approxerror^{2.5}}}$. Moreover,  \fullnonadapalgo{} makes $ \widetilde{O}\fbrac{\frac{\sqrt{\vertexcount}}{\approxerror} }$ \randedgeq{} queries in Line~\ref{alg2: Line 2}. Hence, the total query complexity of \fullnonadapalgo{}  is $\widetilde{O}(\sqrt{n})$.
\end{proof}

%% file: 5_LowerBound.tex
\section{Lower bound}
\label{sec:lowerbound}

We now establish the optimality of our algorithm by proving the following theorem.

\lowerbound*

The lower bound above holds in a stronger query model than that used in our upper bound. The argument combines a hitting bound for local queries with a reduction to support-size estimation for the $\randedgeq$ queries. We begin with the following folklore lower bound for uniformity testing.

\begin{lemma}
\label{lemma:uniform-support-size-estimation}
Let $\mathcal{U}$ be a set of size $n$, and let $\pi:[n]\to \mathcal{U}$ be a uniformly random bijection. Any randomized algorithm that obtains fewer than $\Omega(\sqrt{n})$ samples from a distribution $\mathcal{D}$ over $\mathcal{U}$ cannot distinguish, with success probability at least $2/3$, between the case where $\mathcal{D}$ is uniform over $\mathcal{U}$ and the case where $\mathcal{D}$ is uniform over the subset $\{\pi(1),\pi(2),\dots,\pi(\lfloor n/2\rfloor-1)\}$.
\end{lemma}

This bound follows from the Birthday Paradox; see~\cite[Section~5.1]{Mitzenmacher_Upfal_2005} and~\cite[Page~13]{Canonne2020Survey}. With only $o(\sqrt{n})$ samples, the probability of observing a collision is $o(1)$, so the sample sets are statistically identical in both cases.

We reuse the notation from Lemma~\ref{lemma:uniform-support-size-estimation}. Let $H$ be any graph on $2\lceil \sqrt{n} \rceil + 1$ vertices with at least $n$ edges, and let $f:\mathcal{U}\to E(H)$ be an injective mapping.

We construct a graph $G$ on $n$ vertices as follows. Choose a random subset $S$ of $2\lceil\sqrt{n}\rceil+1$ vertices uniformly at random and identify it with the vertex set of $H$. All vertices outside $S$ are isolated in $G$. For each $x$ in the support of $\mathcal{D}$, include the edge $f(x)$ in $G$, and add no other edges. Thus $|E(G)|$ equals the support size of $\mathcal{D}$.

Consider any non-adaptive algorithm issuing $o(\sqrt{n})$ queries of type $\degreeq$, $\neighbourq$, or $\edgeexistsq$. Since $|S|=2\lceil\sqrt{n}\rceil+1$, each fixed vertex or edge query hits $S$ with probability $O(1/\sqrt{n})$. As the queries are fixed in advance, a union bound shows that, with probability $1-o(1)$, none of them touches $S$. Consequently, all local queries see only isolated vertices and reveal no information about $E(G)$.

Each $\randedgeq$ query, however, returns a uniformly random edge of $G$. Since
\[
E(G)=\{f(x): x\in\operatorname{supp}(\mathcal{D})\},
\]
the distribution of the sampled edge, composed with $f^{-1}$, is exactly $\mathcal{D}$. Thus, distinguishing whether $|E(G)|=n$ or $|E(G)|=\lfloor n/2\rfloor-1$ is equivalent to distinguishing the two cases of Lemma~\ref{lemma:uniform-support-size-estimation}.

By that lemma, any algorithm that succeeds with probability at least $2/3$ must obtain $\Omega(\sqrt{n})$ samples from $\mathcal{D}$, and hence must issue $\Omega(\sqrt{n})$ $\randedgeq$ queries. This completes the proof of Theorem~\ref{theorem:non-adaptive_lower_bound}.